\documentclass[sigconf]{aamas} 

\usepackage{hyperref}

\usepackage{balance} 

\usepackage{comment}
\usepackage{amsmath,amsfonts}
\usepackage{algorithmic}
\usepackage{graphicx}
\usepackage{textcomp}
\usepackage{xcolor}
\usepackage{multirow}
\usepackage{microtype}

\usepackage{amsthm}

\newtheorem{theorem}{Theorem}
\newtheorem{lemma}{Lemma}
\newtheorem{defn}{Definition}
\newtheorem{remark}{Remark}

\setcopyright{ifaamas}
\acmConference[AAMAS '22]{Proc.\@ of the 21st International Conference
on Autonomous Agents and Multiagent Systems (AAMAS 2022)}{May 9--13, 2022}
{Online}{P.~Faliszewski, V.~Mascardi, C.~Pelachaud,
M.E.~Taylor (eds.)}
\copyrightyear{2022}
\acmYear{2022}
\acmDOI{}
\acmPrice{}
\acmISBN{}

\acmSubmissionID{344}

\title[Limit Sets in Multi-Agent Learning]{Poincar\'{e}-Bendixson Limit Sets in Multi-Agent Learning}

\author{Aleksander Czechowski}
\affiliation{
  \institution{Delft University of Technology}
  \country{The Netherlands}
  }
\email{a.t.czechowski@tudelft.nl}

\author{Georgios Piliouras}
\affiliation{
  \institution{Singapore University of Technology and Design}
  \country{Singapore}
  }
\email{georgios@sutd.edu.sg}

\keywords{Replicator Dynamics;
Follow-the-Regularized Leader;
Polymatrix Games;
Poincaré-Bendixson Theorem;
Regret Minimization}

\newcommand{\BibTeX}{\rm B\kern-.05em{\sc i\kern-.025em b}\kern-.08em\TeX}

\begin{abstract}
A key challenge of evolutionary game theory and multi-agent learning is to characterize the limit behavior of game dynamics. Whereas convergence is often a property of learning algorithms in games satisfying a particular reward structure (e.g., zero-sum games), even basic learning models, such as the replicator dynamics, are not guaranteed to converge for general payoffs. Worse yet, chaotic behavior is possible even in rather simple games, such as variants of the Rock-Paper-Scissors game. Although chaotic behavior in  learning dynamics can be precluded by the celebrated Poincar\'e-Bendixson theorem, it is only applicable to low-dimensional settings. Are there other characteristics of a game that can force regularity in the limit sets of learning?
 We show that behavior consistent with the Poincaré-Bendixson
theorem (limit cycles, but no chaotic attractor) can follow purely from the topological structure of the interaction graph, even for high-dimensional settings with an arbitrary number of players and arbitrary payoff matrices. 
We prove our result for a wide class of follow-the-regularized leader (FoReL) dynamics, which generalize replicator dynamics, for binary games characterized interaction graphs where the payoffs of each player are only affected by one other player 
(i.e., interaction graphs of indegree one).
Since chaos occurs already in games with only two players and three strategies, this class of non-chaotic games may be considered maximal. 
Moreover, we provide simple conditions under which such behavior translates
into efficiency guarantees, implying that FoReL learning achieves
time-averaged sum of payoffs at least as good as that of a
Nash equilibrium, thereby connecting the topology of the dynamics to social-welfare analysis.
\end{abstract}
\begin{document}


\pagestyle{fancy}
\fancyhead{}


\maketitle 
\section{Introduction}

Dynamical systems and evolutionary game theory have been instrumental in modern research on multi-agent learning~\cite{bloembergen2015evolutionary, tuyls2005evolutionary, rodrigues2009dynamic, gatti2013efficient, Cesa, Shalev, shoham2008multiagent}.
In particular, characterizing the convergence and limit sets of learning trajectories is vital for understanding the long-term behavior of multi-agent systems.
However, even in simple games, such as Rock-Paper-Scissors~\cite{Sato2,Soda14},
models of evolution and learning are not guaranteed to converge;  even beyond cycles, long-term behavior may lead to chaotic behavior,
known to the dynamical systems community from, e.g., weather models~\cite{lorenz}.
Not only does chaos manifest itself even in simple games with two players, but moreover, a string of recent results suggests that such chaotic, unpredictable behavior may indeed be the norm across a variety of simple low-dimensional game dynamics~\cite{Strien2011,PalaiopanosPP17,benaim2012perturbations,BaileyEC18,BaileyAAMAS19,2020arXiv200513996K,sanders2018prevalence,galla2013complex,frey2013cyclic,chotibut2020family}. Importantly, these results are persistent even for the well-known class of Follow-the-Regularized-leader (FoReL) dynamics~\cite{Mertikopoulos,cheung2019vortices}, despite the fact that FoReL dynamics include some of the most widely studied learning dynamics such as replicator dynamics~\cite{taylor,Hofbauer98}, which is the continuous-time analogue of the Multiplicative Weights Update meta-algorithm~\cite{Arora05themultiplicative}, well known for its optimal regret properties.
Finally, the emergence of chaotic behavior 
 has been connected with increased social inefficiency, which shows that 
 chaotic dynamics can lead to highly inefficient outcomes~\cite{chotibut2019route, Roughgarden}. Such profoundly negative results raise the following questions:

\begin{itemize}
    \item Do simple, robust conditions exist under which learning behaves well?
    \item Which types of games lie at the ``edge of chaos''?
    \item Does dynamic simplicity translate to high-efficiency and social welfare?
\end{itemize}

  Traditionally, a lot of work has focused on showing that, in specific classes of games (e.g., zero-sum or potential games), learning dynamics can lead to convergence and equilibration, see~\cite{Sandholm,Cesa,young2004strategic,fudenberg1998theory} and references therein. 
 Few results span over to general sum games and games of arbitrary payoff structures;
 however, such general approaches are arguably essential in modern research on multi-agent learning.
 For instance, unstructured payoffs can occur naturally when stochastic extensive form games
 are used to create empirical normal form games, by averaging payoffs from simulations for combinations of strategies~\cite{lanctot2017unified, muller2020generalized, wellman2006methods}.
Unstructured payoffs also arise in many real-world applications, such as, e.g., modeling the impact of investing strategies of large funds on the stock market.
While equilibration may not always be possible in such cases, one can still wish to ensure a regularity of sorts in the learning outcomes of the multi-agent system. In particular, 
the famous Poincar\'e--Bendixson theorem (Theorem~\ref{thm:pb}) ensures that two-dimensional continuous learning and adaptation dynamics never form truly chaotic outcomes. However, this comes at a cost: although no specific payoff structure is needed, the underlying learning dynamics must be at most two dimensional.

\paragraph*{Our approach and results.} Rather than by making assumptions on the reward structure or on the dimensionality, we explore a different type of constraint in games. 
We show that the limit behavior of learning
can be determined solely by the topological-combinatorial structure of the game, regardless of the number of players,
or algebraic correlations between the payoffs (e.g., zero-sum).
Firstly, we restrict ourselves to \emph{binary} games~\cite{menezes2006binary, blonski1999anonymous,yu2020computing}, where players have two strategies.
Secondly, we assume that every player can be affected by the behavior of up to one other player. Finally, we add a technical restriction that the game is connected, meaning that it cannot be decomposed into two subgames that are completely independent of each other.
Such games encompass, among others, all $2 \times 2$ games~\cite{goforth2004topology}, Jordan's game~\cite{jordan1993three,gaunersdorfer1995fictitious,hart2003uncoupled}, and easily identifiable subclasses of real-world systems where the graph structure is evident, such as certain traffic networks~\cite{alvarez2010game, Vlassis}, supply chains~\cite{cachon2006game},
or problems of water allocation in deltas~\cite{ambec2002sharing,khmelnitskaya2010values}.
Under these assumptions, we prove in Section~\ref{sec:pb} our main contribution in the form of Theorems~\ref{thm:main} and \ref{thm:main2}, which say that the limit behavior of FoReL learning of these games is always consistent with the 
Poincar\'e--Bendixson theorem.

Having excluded the presence of chaos, we further analyze quantitative properties of binary games, which admit cyclic interaction graphs.
In Section~\ref{sec:ne} we show that, under additional but structurally robust assumptions 
on the payoff matrices 
  (i.e., assumptions that remain valid after small perturbations of the payoff matrices and so are suitable, for example, for empirical payoff matrices), 
one can derive positive results about the efficiency of the time-averaged behavior of the dynamics regardless of whether they are convergent. As is typically the case in the price of anarchy (PoA) literature~\cite{koutsoupias1999worst}, we focus on the measure of \emph{social welfare}, which is the sum of individual payoffs. Whereas the typical PoA literature argues that regret-minimizing dynamics (such as FoReL) are at most a constant factor worse than the behavior of the worst-case Nash equilibrium~\cite{Roughgarden,roughgarden2016twenty}, we instead show that FoReL dynamics are always at least as efficient as the worst-case Nash equilibrium.
Finally, Section~\ref{sec:examples} provides examples of games satisfying our assumptions and their possible limit behavior, as well as a counterexample in the form of a simple binary game that breaks our assumptions and induces chaotic learning dynamics.

\paragraph{Related work.}
First of all, we consider several papers containing complementary results in the form of examples of simple FoReL systems with chaotic dynamics. In addition to the papers mentioned in the introduction, we highlight a chaotic example of Sato et al.~\cite{Sato2} that involves a two-player, three-action game and two complex/chaotic examples in three-player binary games without
 structured interactions~\cite{Plank,Peixe}.
 Comparing these with our assumptions (i.e., binary games and previous-neighbor interactions), we see that
 our results establish a maximal class of games for which such regularity results on limit sets are possible. 

Research that considers non-convergence but focuses on non-chaoticity is scarce. In the closest works to ours, \cite{Nagarajan, NagarajanChaos, Flokas}, the
 authors leverage 
 the Poincar\'{e}--Bendixson theorem to show that the limit behavior of bounded learning trajectories in certain learning systems can be either convergent or cyclic, and in particular no chaotic attractor is possible. 
 However, they do so by assuming low dimensionality (three-player limit) or a nongeneric structure on the set of allowable games, which allows for dimensionality reduction (i.e., a network of $2\times 2$ zero-sum, or coordination games). 
 In terms of connections between cyclic behavior and the efficiency of learning dynamics,~\cite{kleinberg} shows that, for a class of three players, two strategy games with a cyclic attractor can result in social welfare (sum of payoffs) that can be better than the Nash equilibrium payoff; however, the result is once again constrained to the exact game theoretic model.

\section{Preliminaries}

\subsection{Normal form games}

A \emph{finite game in normal form} consists of a set of $N$ players,
each with a finite set of \emph{strategies} $\mathcal{A}_i$.
The preferences of each player are represented by the payoff function
$u_i: \prod_i \mathcal{A}_i \to \mathbb{R}$.
To model the behavior at scale or probabilistic strategy choices, one assumes that players 
use \emph{mixed strategies}, namely, probability distributions
$(x_{i\alpha_i})_{\alpha_i \in \mathcal{A}_i} \in \Delta ( \mathcal{A}_i ) =: \mathcal{X}_i$.
With a slight abuse of notation, the expected payoff of player $i$ in the profile
$(x_{i\alpha_i})_{i,\alpha_i}$ is denoted  $u_i$  and given by
\begin{equation}
  u_i(x) = \Sigma_{\alpha_1 \in \mathcal{A}_1, \dots \alpha_N \in \mathcal{A}_N} u_i(\alpha_1, \dots, \alpha_N) x_{1\alpha_1} \dots  x_{N\alpha_N}.
\end{equation}
A mixed strategy $\hat{x}$ is a \emph{Nash equilibrium}
iff $\forall\ i\text{ and }\forall x: x_{j}=\hat{x}_{j}, \ j \neq i$ we have $u_i(x) \leq u_i(\hat{x})$.
In other words, no player can unilaterally increase their payoff by changing their strategy distribution.
The \emph{minimax value} for player $i$ is given by $\min_{x_{-i}} \max_{x_i} u_i(x)$,
where $x_{-i}:=(x_j)_{j \neq i}$.
This is the smallest possible value that player $i$ can be forced to attain by other players, 
without them knowing the strategy of player $i$.
We call a game \emph{binary} iff $|\mathcal{A}_i|=2$ for all $i$.

\subsection{Graphical polymatrix games}

To model the topology of interactions between players, we restrict our attention to a subset of normal form games, where the structure of interactions between players can be encoded by a graph of two-player normal form subgames, leading us to consider so-called \emph{graphical polymatrix games} (GPGs)~\cite{Kearns, Yanovskaya, Howson}.
A simple directed graph is a pair $(\mathcal{V},\mathcal{E})$, where $\mathcal{V}=\{1,\dots ,N\}$ is a finite
set of \emph{vertices} (representing the players), and $\mathcal{E}$ is a set of ordered vertex pairs
(\emph{edges}), where the first element is called the predecessor, and the second is called the successor.
Each edge $(i,k)$ has an associated two-player normal form game,
where only the successor $k$ is assigned payoffs. These are represented
by a matrix $A^{i,k}$ with rows enumerating the strategies of player $k$, and
columns enumerating the strategies of player $i$.
For a given strategy profile $ s=\{s_i\}_{i} \in \prod_{i} \mathcal{A}_i$,
the payoffs for player $k$ in the full game are then determined 
as the sum 
\begin{equation}
  u_k(s) = \sum_{i: (i,k) \in E} A^{i,k}(s_i,s_k).
\end{equation}
The payoffs can be extended to mixed strategies in a standard multilinear fashion:
\begin{equation}
  u_k(x) = \sum_{i: (i,k) \in E} \sum_{x_{s_i}, x_{s_k}} A^{i,k}(s_i,s_k)x_{s_i} x_{s_k}.
\end{equation}

A situation where both the successor $k$ and  the predecessor $i$
obtain a reward can be modeled by including both edges $(i,k)$ and $(k,i)$ in the graph.

We say that a simple directed graph is weakly connected 
if any two vertices can be connected by a set of edges, where the direction of the edges
is not considered.
This is a weaker condition than strong connectedness,
where each pair of vertices must be connected by a \emph{path}
(i.e., a
sequence of edges together with associated vertices, where the successor in one edge is the predecessor in the next).
The \emph{indegree} of a vertex
is the number of edges for which the vertex is the successor (i.e.,
the number of  predecessors). The \emph{outdegree} is the number of edges for which the vertex is the predecessor (i.e. the number of  successors).
A \emph{cycle} is a path where the predecessor
in the first edge is the successor in the last edge.
For our exposition we  identify cycles modulo shifts, i.e., if two paths consist of the same edges in shifted order, then they form the same cycle. 
In this paper we consider two types of weakly connected GPGs: 

\begin{enumerate}
\item First, \emph{cyclic} games, where the interaction between the players forms a cycle, where each player interacts only with the previous neighbor. We observe that in such a cyclic game the indegree and outdegree of each vertex is one. For simplicity, we  label the nodes of such $N$-player games by natural numbers $i=0,1,\dots,N$
and use the convention that node $i$ is the successor to node $i-1$, and that node $0$ is identified with node $N$.
\item Second, a more general class of graphical games, where each player's payoffs depend on at most one other player (i.e., the indegree of each vertex is at most one). For a vertex $i \in \mathcal{V}$, we   denote the predecessor vertex by $\hat{i}$, if it exists. For cyclic games we have $\hat{i}=i-1$.
\end{enumerate}

Below, we state and prove a simple lemma that characterizes the one-predecessor assumption in terms of graph topology and clarifies the relation between cyclic and indegree-one graphs (cf. Figure~\ref{graphfigure}).

\begin{lemma}\label{graphlemma}
Let $(\mathcal{V},\mathcal{E})$ be a weakly connected, simple, directed graph.
If the indegree of each vertex is at most one,
then the graph can have at most one cycle. 
If the graph has no cycle, then it has  at most one root vertex (i.e., a vertex of indegree zero), such that all other vertices are connected to it by a unique directed path.
\end{lemma}

\begin{proof}

For the first part of the lemma, we assume the contrary: that $a_1$, $a_2$ are nodes of two distinct cycles within the same weakly connected component.
The edges between $a_1$ and $a_2$ must form a path
(otherwise there would be a vertex with two predecessors).
Assume the path leads from $a_1$ to $a_2$
and let $a_0$ be the first vertex which is both on the path and on the cycle of $a_2$.
Then $a_0$ has two predecessors, which leads to a contradiction.

For the second part of the lemma we argue as follows. If any vertex  has a sequence of predecessors that does not form a cycle, and does not have a root node, then by backtracking through the predecessors we could identify an infinite collection of distinct vertices. Therefore, there must be at least one root node for each vertex. The path from such a root node to the given vertex must be unique, otherwise one could identify a vertex along the path with two predecessors. Finally, it is impossible to have two distinct root nodes, as  connectedness imposes that there would have to exist a node with two predecessors between them. 
\end{proof}

\begin{remark}\label{graphremark}
Under the assumptions of Lemma~\ref{graphlemma}, if the graph has a cycle,
then the cycle enjoys  properties similar to those of a root node:
 no paths go from outside the cycle to the cycle (otherwise one vertex in the cycle would have two predecessors),
and all vertices outside  the cycle must be connected by a path from one of the vertices of the cycle
(a unique path, up to the starting point within the cycle).
Later, we shall refer to such cycle as the \emph{root cycle}.
\end{remark}

\begin{figure}
    \includegraphics[width=70mm,trim=0 200 0 200,clip]{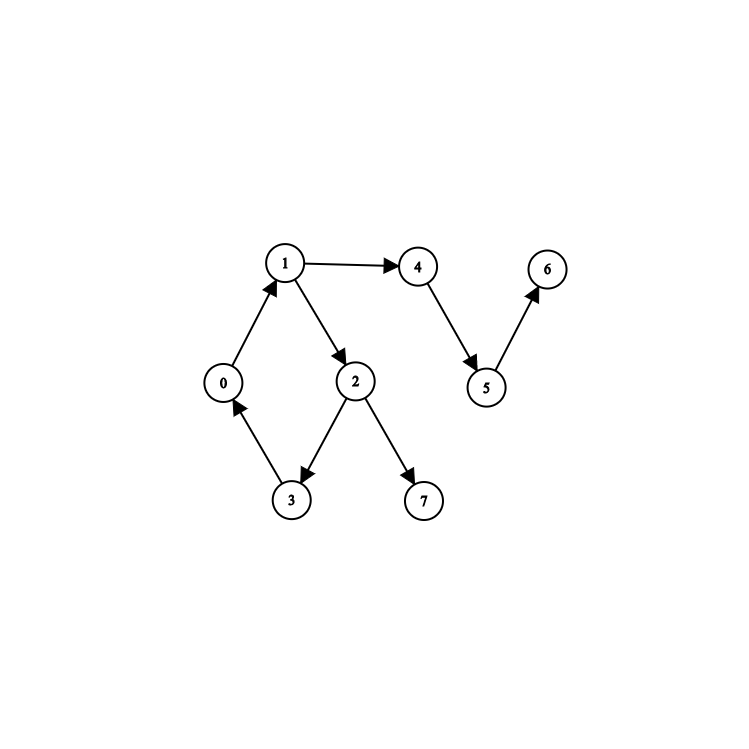}
    \caption{A weakly connected graph where each vertex is at most of indegree  one.}
    \label{graphfigure}
\end{figure}

\subsection{Follow-the-regularized-leader equations}

Denote by $v_{i\alpha_i}(x):=u_i(\alpha_i;x_{-i})$ and $v_i(x)=(v_{i\alpha_i}(x))_{\alpha_i \in \mathcal{A}_i}$. To model the dynamics of learning we use a class of learning systems
known as \emph{follow-the-regularized-leader} systems (FoReL)~\cite{Cesa,Shalev}.
This class encompasses a variety of models ranging from gradient to replicator dynamics, and allows for natural description of learning
as regularized maximization of individual payoffs.

FoReL dynamics for player $i$ are defined by evolution of 
\emph{utilities} $y_i = \{y_{i\alpha_i}\}_{\alpha_i \in \mathcal{A}_i} \in \mathbb{R}^{|\mathcal{A}_i|}$ -- that is real numbers representing a score each player assigns to each respective strategy -- by the integral equation
\begin{equation}\label{forel}
\begin{aligned}
  y_i (t) &= y_i(0) + \int_0^t v_i(x(s)) ds,\\
  x_i(t)&=Q_i(y_i(t)),
\end{aligned}
\end{equation}
where the \emph{choice map} $Q=(Q_1,\dots,Q_N)$, $Q_i: \mathbb{R}^{|\mathcal{A}_i|} \to \mathcal{X}_i$, which determines the evaluated strategy profile $x(t)$ is
given on each coordinate by:
\begin{equation}
  Q_i(y_i) = \operatorname{argmax}_{x_i \in \mathcal{X}_i} \{ \langle y_i, x_i \rangle - h_i(x_i) \> \}.
\end{equation}
In the above $h_i: \mathcal{X}_i \to \mathbb{R} \cup \{ -\infty, \infty \}$ 
is a convex regularizer function, representing a regularization/exploration term.
The equation~\eqref{forel} represents how players adapt their mixed strategies to changing 
utility values.
Observe, that without the regularization term, the map $Q_i$ would simply put all weight on the strategy with the highest utility.

In binary games, each player has only two strategies at his disposal, say $\alpha_0,\alpha_1$.
The variable $x_i$ denotes then the proportion of time player $i$ plays strategy $\alpha_0$,
and the proportion of $\alpha_1$ is given by $1-x_i$. 
Following~\cite{Mertikopoulos}, we introduce new variables 
$z_{i} := y_{i\alpha_0} - y_{i\alpha_1} \in \mathbb{R}$,
representing the difference in utilities between playing strategy $\alpha_0$ and $\alpha_1$. 
It is intuitively clear, and it was proved formally e.g. in~\cite{Mertikopoulos} that $Q_i(z_i+c,c)$ is constant in $c$, and
therefore, without loss of generality, we can set $c:=0$,
and restrict our considerations to a $z$-dependent choice map $\hat{Q}_i(z_i):=Q_i(z_i,0)$.
Provided that $Q$ is sufficiently regular (e.g. continuous), the integral equation~\eqref{forel} can be converted to a system of differential equations
\begin{equation}\label{deq}
\dot{z}=V(z),
\end{equation}
given coordinate-wise by
\begin{equation}\label{vf}
  V_{i}(z):= v_{i\alpha_0}(\hat{Q}(z))-v_{i\alpha_1}(\hat{Q}(z)),
\end{equation}
for details again see~\cite{Mertikopoulos}.

\begin{remark}\label{auton}
An intuitively obvious, but technically important observation
is that evolution of $i$th coordinates of the system~\eqref{forel}, and, in turn~\eqref{vf}
depends solely on the values of $x_j$ or $z_j$, respectively, 
for nodes $j$ that influence the payoffs of $i$.
In particular, for GPGs we have $\partial V_i / \partial z_j \neq 0$
implies that there is an edge from $j$ to $i$ in the game graph;
and for GPGs with up to one predecessor, without loss of generality we can
rewrite~\eqref{deq} as
\begin{equation}
    \dot{z}_i = V_i(z_{\hat{i}}) =v_{i\alpha_0}(\hat{Q}_{\hat{i}} (z_{\hat{i}}))-v_{i\alpha_1}(\hat{Q}_{\hat{i}}(z_{\hat{i}})).
\end{equation}
\end{remark}

As previously hinted, for equation~\eqref{vf} to be well-posed,
we need to enforce certain conditions on the regularizer.
The following lemma determines desirable properties of monotonicity and smoothness of the choice map, when a player has exactly two strategies at disposal (so $\mathcal{X}_i=[0,1]$).

\begin{lemma}\label{induclem}
  Assume that the regularizer $h_i$ satisfies the following conditions:
  \begin{enumerate}
    \item $h_i \in C^2( (0,1)) \cap C^0([0,1])$ (\emph{smoothness}),
    \item $h_i'(x_i) \to -\infty$ as $x_i \to 0$ and $h_i'(x_i) \to \infty$ as $x_i \to 1$ (\emph{steepness}),
    \item $h_i''(x_i) > 0$ for $x \in (0,1)$ (\emph{strict convexivity}).
  \end{enumerate}
  Then $\hat{Q}_i \in C^1(\mathbb{R})$ and $\hat{Q}_i'(z_i)>0$.
\end{lemma}

\begin{proof}
For a given $z_i$, 
$\hat{Q}_i(z_i)$ is defined as the maximizer of $\langle (z_i,0), (x_i, 1-x_i) \rangle - h_i(x_i)$
over $x_i \in [0,1]$. We have
\begin{equation}
\langle (z_i,0), (x_i, 1-x_i) \rangle - h_i(x_i) = z_i x_i - h_i(x_i).
\end{equation}
From steepness, continuity and strict convexity it follows that $h_i(0)=h_i(1)=\infty$
so the maximum cannot be attained there.
A necessary condition for maximum to be attained in $(0,1)$
is
\begin{equation}\label{nec}
 z_i=h_i'(x_i).
\end{equation}

From steepness and strict convexivity 
it follows that equation~\eqref{nec} has a unique solution $x_i=:\hat{Q}_i(z_i)$
for any $z_i \in \mathbb{R}$.
From the inverse function theorem we have
\begin{equation}
  \frac{ \partial x_{i} }{\partial z_i } = \hat{Q}_i'(z_i) = 1/h_i''(x_i) > 0,
\end{equation}
which also implies that $\hat{Q}_i$ is $C^1$.
\end{proof}

Perhaps the best known example of a FoReL
learning system are the \emph{replicator equations}~\cite{taylor}, where the regularizer is given by
\begin{equation}
    h_i(x_i):=\sum_{\alpha_i} x_{i\alpha_i}\log{x_{i\alpha_i}}.
\end{equation}
In particular, such regularizer satisfies the assumptions of Lemma~\ref{induclem}, and yields the following equations for a binary GPG with up to one predecessor:
\begin{align}\label{replicator}
  \begin{aligned}
  &\dot{z}_i=
  \sum_{j,k\in\{0,1\}} (-1)^{(j+k)}A^{\hat{i},i}(\alpha_j, \alpha_k)\frac{\exp(z_{\hat{i}})}{1+\exp(z_{\hat{i}})}
  \\
  &- A^{\hat{i},i}(\alpha_1,\alpha_1) + A^{\hat{i},i}(\alpha_1,\alpha_0),\quad i = 1,\dots , N
\end{aligned}
\end{align}
which translates to the following system in original ($x$) coordinates:
\begin{align}\label{replicatorX}
  \begin{aligned}
  &\dot{x}_i=x_i(1-x_i)\sum_{j,k\in\{0,1\}} (-1)^{(j+k)}A^{\hat{i},i}(\alpha_j, \alpha_k)x_{\hat{i}}\\ &-x_i(1-x_i)\left( A^{\hat{i},i}(\alpha_1,\alpha_1) - A^{\hat{i},i}(\alpha_1,\alpha_0)\right),\quad i=1,\dots ,N.
  \end{aligned}
  \end{align}
%

\subsection{Limit sets, periodic orbits and chaos}

A differential equation $\dot{x}=F(x)$ given by a $C^1$ vector field 
$F: \Omega \to \mathbb{R}^n$ on a domain $\Omega \subset \mathbb{R}^n$
admits a unique solution on a maximal open interval $I=(I_l,I_r),\ I_l,I_r \in \mathbb{R} \cup \{\pm \infty\}$, denoted by 
$x(t): I \to \mathbb{R}^n$, 
for any initial condition $x(0) = x_0 \in \Omega$.
Among possible solutions to such equation, we distinguish particular types of solutions
defined by their qualitative properties:
we say that a solution $x(t)$ is an \emph{equilibrium} iff $x(t)=\operatorname{const}$ for all $t \in I$.
A solution is \emph{periodic} iff $x(t)=x(t+T)$ for some $T>0$ and all $t \in I$;
and it is a \emph{connecting orbit} between equilibria $x_1$ and $x_2$ (allowing $x_1=x_2$),
iff $x(t) \to x_1$ as $t \to \infty$ and $x(t) \to x_2$ as $t \to -\infty$.
A set $\omega(x_0) \subset \Omega$ is a limit set for an initial condition $x_0 \in \Omega$,
if $\forall x \in \omega(x_0)$ there exists an unbounded, increasing sequence $\{t_n\}_n \subset \mathbb{R}^+$,
such that $x(t_n) \to x, \ n \to \infty$.
Limit sets are \emph{invariant} -- they are formed by unions of solutions 
of the differential equation on maximal intervals. They are also \emph{compact} -- bounded as subsets of $\mathbb{R}^n$,
and closed under the limit operation on sequences from itself.

Fundamental research has been devoted
to study the properties of solutions within limit sets,
as they offer a qualitative description of long-term behavior of the system~\cite{hale}.
Since the discovery of chaotic attractors~\cite{lorenz}, it has become known
that in the general setting, these solutions can have arbitrarily complicated shapes and exhibit
seemingly random behavior,
a clearly undesirable feature from the point of view of applications;
and engineering systems with simple $\omega$-limit sets became of particular interest.
\begin{defn}
  We say that a differential equation $\dot{x}=F(x),\ x \in \Omega$ has the Poincar\'e-Bendixson property
  iff for all $x \in \Omega$, such that the solution $x(t)$ is bounded, each limit set $\omega(x)$ such that $\omega(x)\subset \Omega$ is either:
  \begin{itemize}
    \item an equilibrium;
    \item a periodic solution;
    \item a union of equilibria and connecting orbits between these equilibria.
  \end{itemize}
\end{defn}

A well known result from the qualitative theory of differential equations
shows that planar systems exhibit this trait.
\begin{theorem}\label{thm:pb}  The Poincar\'e-Bendixson Theorem~\cite{Bendixson}.
  Let $F=F(x)$, $x \in \Omega \subset \mathbb{R}^2$ be a $C^1$ vector field with finitely many zeroes. 
  Then, the differential equation $\dot{x}=F(x)$ has the Poincar\'e-Bendixson property.
\end{theorem}

Already in $\mathbb{R}^3$ there are known examples of systems having complicated, chaotic attractors~\cite{lorenz}. 
However, dimensionality is not the only factor which could determine potential shapes of limit sets.
In particular, for certain systems of arbitrary dimension, with structured
``previous-neighbor'' interactions between the variables, the limit sets can be as as simple
as in planar systems.

\begin{theorem} Mallet-Paret \& Smith\label{paretsmith}~\cite{ParetSmith}.
 Let $x=(x_1, \dots, x_n)$, $(f_i(x_{i-1},x_i))_{i=1}^n$, 
 be a $C^1$ vector field on an open, convex set $O \subset \mathbb{R}^n$, and let $x_0:=x_n$.
 Assume that $\frac{\partial f_{i}}{\partial x_{i-1}} \neq 0$ for all $x \in O$. Then, the system of differential equations
 \begin{equation}
   \dot{x}_i = f_i(x_{i-1},x_i), \ i=1,\dots ,n,\  x \in O,
 \end{equation}
 has the Poincar\'e-Bendixson property.
\end{theorem}
The above theorem is key to proving our further results.

\section{The Poincar\'e-Bendixson theorem for games}\label{sec:pb}

In this section we state and prove our main results on the topology of limit sets in Follow-the-regularized-Leader learning.
We will first state and prove the Poincar\'e-Bendixson theorem for cyclic games:

\begin{theorem}\label{thm:main}
  Let $\dot{z}=V(z)$ be a system of differential equations given by the vector field~\eqref{vf} -- the 
  follow-the-regularized-leader learning dynamics -- for a binary, cyclic game.
  For any smooth, steep, strictly convex collection of regularizers $\{h_i\}_i$  such system
  possesses the Poincar\'e-Bendixson property.
\end{theorem}

\begin{proof}

Since $u_i$ depends only on $Q_{i}$ and $Q_{i-1}$, we have
\begin{align}
\begin{aligned}
  V_{i}(\hat{Q}(z)) &= V_{i}(\hat{Q}_{i-1}(z_{i-1}))\\ 
                    &= v_{i\alpha_0}(Q_{i-1}(z_{i-1},0))-v_{i\alpha_1}(Q_{i-1}(z_{i-1},0)).
\end{aligned}
\end{align}

Our goal is to employ Theorem~\ref{paretsmith}. 
Therefore, we would like to establish under which conditions
\begin{equation}\label{ineq}
 \frac{\partial V_i}{\partial z_{i-1}} \neq 0.
\end{equation}
for all $i$.
We have:
\begin{equation}
  \frac{ \partial V_i }{ \partial z_{i-1} } = \frac{\partial v_{i\alpha_0}}{\partial x_{i-1}} \frac{\partial x_{i-1}}{\partial z_{i-1}} 
  - 
 \frac{\partial v_{i\alpha_1}}{\partial x_{i-1}} \frac{\partial x_{i-1}}{\partial z_{i-1}} .
\end{equation} 

Moreover, differentiation of mixed strategy payoffs yields
\begin{align}
    \begin{aligned}
    \frac{\partial v_{i\alpha_1}}{\partial x_{i-1}} - \frac{\partial v_{i\alpha_0}}{\partial x_{i-1}}
= A^{\hat{i},i}(\alpha_0, \alpha_{0})- A^{\hat{i},i}(\alpha_1, \alpha_{0})\\
+A^{\hat{i},i}(\alpha_1, \alpha_1)
-A^{\hat{i},i}(\alpha_0, \alpha_1).
    \end{aligned}
\end{align}

From Lemma~\ref{induclem} we have $\frac{\partial x_{i-1}}{\partial z_{i-1}} > 0$, 
so the necessary condition to satisfy inequality~\eqref{ineq} is:
\begin{equation}
\begin{aligned}
&A^{\hat{i},i}(\alpha_0, \alpha_{1}) +A^{\hat{i},i}(\alpha_1, \alpha_{0}) \\ &\neq
A^{\hat{i}, i}(\alpha_0, \alpha_{0}) +  A^{\hat{i}, i}(\alpha_1, \alpha_{1}).
\end{aligned}
\end{equation}

Now let's consider the edge case, where
$A^{\hat{i},i}(\alpha_0, \alpha_{1}) +A^{\hat{i},i}(\alpha_1, \alpha_{0})  =
A^{\hat{i}, i}(\alpha_0, \alpha_{0}) +  A^{\hat{i}, i}(\alpha_1, \alpha_{1})$ for some $i$. 
Then $ \partial v_{i\alpha_0}/\partial x_{i-1} = \partial v_{i\alpha_1}/{\partial x_{i-1}}$.
Consequently, $\partial V_i / \partial z_{i-1} =0$,
and hence $i$-th coordinate of all solutions has the form $z_i(t)=a_i t+b$, for some $a_i,b_i$.
If $a_i \neq 0$, then all solutions diverge to infinity.
If, however $a_i=0$, then $z_i(t)=const$.
Since $V_{i+1}$ depends only on $z_i$, 
and $z_{i+1} = a_{i+1}t+b_{i+1}$;
the argument continues, until all coordinates of solutions are constant, or one coordinate diverges for all solutions.
\end{proof}

We are now ready to state and prove the theorem for GPGs with nodes of indegree at most one. 

\begin{theorem}\label{thm:main2}
Let $\dot{z}=V(z)$ be a system of differential equations given by the follow-the-regularized leader dynamics of a binary, weakly connected, graphical polymatrix game, where each player has up to one predecessor.
Then, for any smooth, steep, strictly convex collection of regularizers $\{h_i\}_i$, such system possesses the Poincar\'e-Bendixson property.
\end{theorem}

First, we state the following lemma on inheritance of the Poincar\'e Bendixson property for augmented systems.

\begin{lemma}\label{auglemma}
Consider the following $y$-augmented system of differential equations 
\begin{equation}\label{aug}
\begin{aligned}
\dot{x}&=f(x),\\
\dot{y}&=g(x_i),\\
x&=\{x_1,\dots,x_n\} \in \mathbb{R}^n, \ y \in \mathbb{R}.
\end{aligned}
\end{equation}
for smooth $f$, $g$.
If the original system 
\begin{equation}\label{orig}
    \dot{x}=f(x)
\end{equation} 
has the Poincar\'e-Bendixson property, then the augmented system~\eqref{aug} also has the Poincar\'e-Bendixson property.
\end{lemma}

\begin{proof}
Let $Z$ be an $\omega$-limit set corresponding to some solution $(x(t),y(t))$ to the system~\eqref{aug}.
Consider $X$ -- an $\omega$-limit set to solution $x(t)$ of~\eqref{orig}.

From invariance of $\omega$-limit sets it follows set $Z$ consists of a union of solutions of~\eqref{aug}.
For any solution $\{x^*(t), y^*(t): \ t \in \mathbb{R}\} \subset Z$, we have $\{x^*(t)\} \subset X$.
By the Poincar\'e-Bendixson property of the original system,
we can distinguish three cases:
\begin{enumerate}
    \item $x^*(t)$ is an equilibrium of~\eqref{orig}, \label{en1} 
    \item $x^*(t)$ is a periodic orbit of~\eqref{orig}, \label{en2}
    \item $x^*(t)$ is a connecting orbit of~\eqref{orig} -- a part of a cycle of connecting orbits. \label{en3}
\end{enumerate}

In the rest of the proof we will frequently use the integral form of solutions $y(t)$ to~\eqref{aug}, given by $y(t)=y(0)+\int_0^t g(x_i(s))ds$.

Case~\eqref{en1}: We  prove that $(x^*(t), y^*(t))$ is stationary for~\eqref{aug}. It is enough to show $g(x^*_i)=0$. Assume otherwise. Then $|y^*(t)|=|y(0)+\int_0^t g(x^*_i)ds|=|y(0)+tg(x^*_i)| \to \infty$ as $t \to \pm \infty$. This contradicts the boundedness of an $\omega$-limit set.

Case~\eqref{en2} Let $T$ be the period of $x^*(t)$. We show that $(x^*(t),y^*(t))$ is a periodic solution of~\eqref{aug} of the same period. We have: 
\begin{equation}
\begin{aligned}
    \frac{d}{dt}(y^*(t+T)-y^*(t))&=\frac{d}{dt}\int_{t}^{T+t}g(x^*_i(s))ds\\
    &=g(x^*_i(T+t))-g(x^*_i(t))\\
    &=0,
\end{aligned}
\end{equation} hence $y^*(t+T)-y^*(t)=const$. If this quantity would be non-zero, the diameter of the set $\{y^*(t):\ t \in \mathbb{R}\}$ would be infinite. However, the set $Z$ is bounded, and therefore $y^*(t+T)=y^*(t)$.

Case~\eqref{en3}: We  show that $(x^*(t),y^*(t))$ is a connecting orbit between two equilibria for the full system~\eqref{aug}. We shall only prove convergence with $t \to \infty$, the very same argument holds for $t \to -\infty$ and $\alpha$-limit sets.
The orbit $(x^*(t),y^*(t))$ is bounded and therefore it has an accumulation point as $t \to \infty$ given by $(x^{**},y^{**}) \in \omega{(x^{*}(0),y^{*}(0))}$. The point $x^{**}$ is an equilibrium for~\eqref{orig}. We will show that $(x^{**},y^{**})$ is an equilibrium. It is enough to show that $g(y^{**})=0$. Assume otherwise.
Then $y^{**}(t)=y^{**}+tg(x^{**}_i)$ which is unbounded. However, it is also a part of  $\omega{( (x^{*}(0),y^{*}(0)) )}$, since $\omega$-limit sets are invariant. Boundedness of $\omega{ ((x^{*}(0),y^{*}(0)) )}$ leads to a contradiction. 
The same process, repeated for all connecting orbits of~\eqref{orig}, 
creates a cycle of connecting orbits for~\eqref{aug}.
\end{proof}

 Now, we can proceed to the proof of Theorem~\ref{thm:main2}.

\begin{proof}
By Lemma~\ref{graphlemma}, and Remark~\ref{graphremark}, we know that the graph of the system has either a root vertex or a root cycle. We will first address the case of a root vertex. We will see that this case is somewhat degenerate. Without loss of generality let us assume that it is labelled as the 1st vertex, and that the other vertices are numbered in order of increasing path distance from vertex 1 (i.e. $j<i$ implies that the path from $1$ to $j$ is shorter than the path from $1$ to $i$) -- this is possible by Lemma~\ref{graphlemma}.

The payoffs of the root node are only affected by its own choice of strategy. Therefore, we can write $\dot{z}_1=u_1(\alpha_0)-u_1(\alpha_1)$, and, consequently, $z_1(t)=t(u_1(\alpha_0)-u_1(\alpha_1))+z_1(0)$.
This system constitutes an autonomous ODE, which trivially has the Poincar\'e-Bendixson property (as it is either completely stationary, or is divergent). 
From then on, we can add nodes, starting from vertices connected to the root vertex, and then continuing in an inductive fashion. Then, either one of the nodes diverges, or they are all stationary, and trivially satisfy the Poincar\'e-Bendixson property.
It should be noted that "divergence" in practice means that $z_i(t)$'s approach in the limit $t \to \infty$ to either $\infty$ or $-\infty$; the former implies that the player $i$ is placing almost all probability mass on strategy $\alpha_0$, and the latter -- on $\alpha_1$.

The more interesting scenario arises for the root cycle, where periodic limit sets are possible.
Enumerate these vertices by $1,\dots,N_0$, with $N_0\leq N$, 
and assume that the vertices from $N_0+1$ to $N$ are arranged in the order of increasing path distance from vertices of the cycle (possible by Remark~\ref{graphlemma}).
Observe that the system 
\begin{equation}
\begin{aligned}
\dot{z}_i&=V_i(z_{\hat{i}}),\\ 
i&=1,\dots,N_0,
\end{aligned}
\end{equation}
is an autonomous system of differential equations (as there are no edges with successors in $\{1,\dots,N_0\}$, and predecessors outside of this set), and forms a binary, cyclic game in the sense of Theorem~\ref{thm:main}.
As such, this subsystem possesses the Poincar\'e-Bendixson property. From then on, the proof continues similarly as for the root vertex. We add a vertex $N_0+1$ which has an incoming edge from the root cycle, and, by Lemma~\ref{auglemma} observe that the system  
\begin{equation}\label{next}
\begin{aligned}
\dot{z}_i&=V_i(z_{\hat{i}}),\\ 
i&= 1,\dots,N_0+1,
\end{aligned}
\end{equation}
again has the Poincar\'e-Bendixson property. The proof continues inductively w.r.to the vertices,
until we conclude that the full system $\dot{z}=V(z)$ has the Poincar\'e-Bendixson property.
\end{proof}

\begin{remark}
Theorems~\ref{thm:main},~\ref{thm:main2} apply to dynamics of fully mixed initial strategy profiles bounded away from pure strategies, as FoReL learning~\eqref{forel} is ill-defined for pure strategies.
For some learning models such as as the replicator equations~\eqref{replicatorX}~the theorems can be applied to subsystems arising when certain players assume a pure strategy profile, as in these models pure strategy profiles define invariant learning spaces.
\end{remark}

\section{From Geometry to Efficiency: Social Welfare Analysis}\label{sec:ne}

The following result shows that 
for cyclic, binary games, under additional but structurally robust assumptions 
on the payoff matrices 
   (i.e., assumptions that remain valid after small perturbations of the payoff matrices), the time-average social welfare of our FoReL dynamics is at least as high, as the social welfare $SW=\sum_i u_i$ of the worst Nash equilibrium. The proof crucially relies on the interplay of the optimal regret properties of FoReL dynamics combined with structural characterizations of the set of Nash equilibria of these games.

\begin{theorem}\label{thm:minimax}
In any  binary, cyclic game with the property that for any player $i$, the payoff entries are distinct and $$[A^{i-1,i}(\alpha_0,\alpha_0)-A^{i-1,i}(\alpha_1,\alpha_0)][A^{i-1,i}(\alpha_0,\alpha_1)-A^{i-1,i}(\alpha_1,\alpha_1)]<0,$$
the time-average of the social welfare of FoReL dynamics 
is at least that of the social welfare of the worst Nash equilibrium.
Formally,
\begin{equation}\lim\inf \frac{1}{T}\int_0^T \sum_i u_i(x(t)) dt \geq \sum_i u_i(x_{NE}),    
\end{equation}

\noindent
where $x_{NE}$ the worst case Nash equilibrium, i.e., a Nash equilibrium that minimizes the sum of utilities of all players.
\end{theorem}

In other words, the Nash equilibrium is the worst imaginable outcome for all players; and the dynamical, regret minimization approach yields superior payoffs.
\begin{proof}
Lets consider the payoff matrix of  each player $i$. Recall, that by the cyclicity assumption, 
there is at most one player $k$ such that $A^{k,i}$ is a non-zero matrix, i.e., the unique predecessor of $i$, that for simplicity of notation we call $i-1$. 
By assumption, the four entries will be considered distinct. Next, we break down the analysis into two cases. As a first case, we consider the scenario where there exists at least one player with a strictly dominant strategy. The FoReL dynamics of that player strategy profile will trivially converge to playing the strictly dominant strategy with probability one. Similarly,  all players reachable from player $i$ will similarly best respond to it. This is clearly the unique NE for the binary cyclic game, so in this case the limit behavior of FoReL dynamics exactly corresponds to the unique Nash behavior and the theorem follows immediately.

Next, let's consider the case where no player has a strictly dominant strategy. In this case, we will construct a specific Nash equilibrium for the cyclic game (although it may have more than one). In this Nash equilibrium every player $i-1$ plays the unique mixed strategy that makes its successor (player $i$) indifferent between its two strategies. Such a strategy exists for each player, because otherwise there would exist a player with a strictly dominant strategy. In fact by the assumption $[A^{i-1,i}(\alpha_0,\alpha_0)-A^{i-1,i}(\alpha_1,\alpha_0)][A^{i-1,i}(\alpha_0,\alpha_1)-A^{i-1,i}(\alpha_1,\alpha_1)]<0$ such a strategy would be the $i-1$st player's min-max strategy if they participated in a zero-sum game with player $i$ defined by the payoff matrix of player $i$.
 Indeed, this assumption, along with the fact that player $i$ does not have a dominant strategy, exactly encodes that the zero-sum game  (defined by payoff matrix $A^{i-1,i}$) has an interior Nash.  
Given its predecessors behavior, player $i$ will be receiving exactly its max-min payoff no matter which strategy they select, therefore this strategy profile where each player $i-1$ just plays the strategy that makes player $i$ indifferent between their two options is a Nash equilibrium, where each player receives exactly their max-min payoffs. 
However, by~\cite{Mertikopoulos} (Lemma C.1), continuous-time FoReL dynamics are no-regret with their time-average regret converging to zero at an optimal rate of O(1/T), i.e. there exists an $\Omega_i>0$,
such that for all players $i$ we have:
\begin{equation}
    \max\limits_{p_i \in \mathcal{X}_i} \frac{1}{T} \int_0^T 
    \left(u_i(p_i; x_{-i}(t))-u_i(x(t)\right) dt \leq \frac{\Omega_i}{T}.
\end{equation}
However, the left hand side is greater or equal to \begin{equation}u_i(x_{NE})-\frac{1}{T} \int_0^T
    u_i(x_i(t))),\end{equation} since the mixed Nash equilibrium consists of max-min strategies.
Therefore, the sum over $i$ of the time-average performance is at least the sum of the max-min utilities minus a quickly vanishing term O(1/T) and the theorem follows. 
\end{proof}

\section{Examples}\label{sec:examples}

To illustrate our theoretical results,
we analyze the replicator dynamics~\eqref{replicatorX} of two classes multidimensional binary cyclic games that exhibit non-convergence and therefore non-trivial limit behavior. 
The goal of the examples is to show that all possible limit sets indicated in the Poincar\'e--Bendixson property (i.e., an equilibrium, a periodic solution, and a cycle of connecting solutions) are attainable for systems satisfying our assumptions. In addition, we plot the social welfare of simulated trajectories, relating them to the results of Theorem~\ref{thm:minimax}. Finally, we provide a counterexample in the form of a three-dimensional replicator system that violates the assumptions of our theorems and exhibits chaos.
To determine the limit sets, we  numerically integrate the initial-value problems with various starting conditions
via the \emph{lsoda} differential equation integrator~\cite{hindmarsh2005lsoda}.

\subsection{Matched-mismatched pennies game}

First, we analyze a four-dimensional game of matched-mismatched pennies.
Each player has a choice of two strategies, $\alpha_0$ and $\alpha_1$.
The payoffs for players 0 and 2 are given by 
\begin{equation}
  A^{3,0}=A^{1,2}=
\begin{bmatrix}
-1 & 1 \\
1 & -1 
\end{bmatrix}
\end{equation}
and the payoffs for players 1 and 3 are given by 
\begin{equation}
  A^{0,1}=A^{2,3}=
\begin{bmatrix}
1 & -1 \\
-1 & 1 
\end{bmatrix}.
\end{equation}
Simply put, players $0$ and $2$  try to mismatch the strategy with players $1$ and $3$,
and players $1$ and $3$  try to match them.

The system possesses three Nash equilibria, which correspond to the
following strategy profiles: $(0,0,1,1)$, $(1,1,0,0)$,
$(0.5,0.5,0.5,0.5)$, out of which the pure Nash equilibria are attracting,
and the mixed Nash equilibrium has two center directions: one repelling 
and one attracting. We  denote the mixed Nash equilibrium by $x_{MNE}$. 
Given the symmetry of the system, the plane 
$\{(t,s,t,s), \ t,s \in [0,1]\}$
is invariant,
 consists purely of periodic orbits, and forms the 
center manifold to the mixed Nash equilibrium.

\begin{figure}
    \includegraphics[width=40mm, trim= 60 0 60 0]{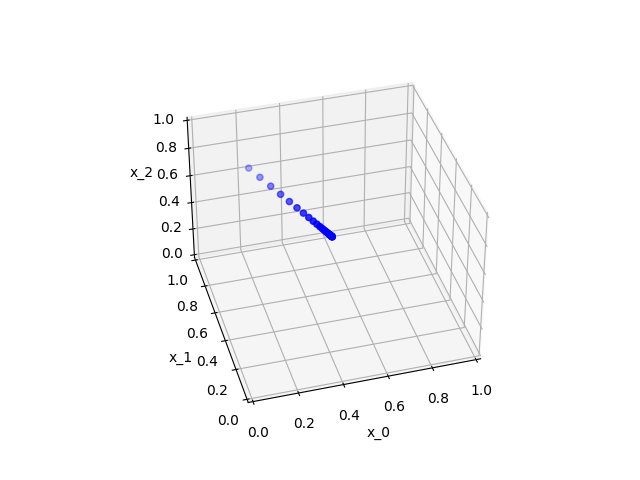}
      \includegraphics[width=40mm, trim= 60 30 80 30]{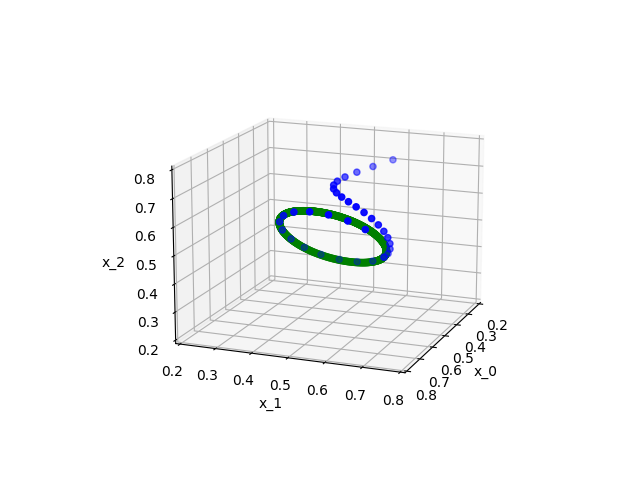}
  \caption{Limit sets in the matched-mismatched pennies system: an orbit converging to an equilibrium (left) and an orbit converging to a limit cycle (right).
  }\label{Fig1}
\end{figure}

The numerical results are consistent with Theorems~\ref{thm:main} and~\ref{thm:main2}. 
The only limit sets observed by the numerical simulations are the mixed Nash equilibrium $x_{MNE}$ itself (along a single-dimensional attracting set) and the limit cycles around it, which also appear to be of saddle nature and have a single attracting direction, see Figure~\ref{Fig1}.
Most crucially, more complicated behavior, such as chaos or invariant tori, does not emerge, despite the system being nontrivially embedded in four dimensions. 

The mixed Nash equilibrium
yields the minimax payoff vector $(0,0,0,0)$ for each player and the social welfare of $0$.
The payoff matrices satisfy the assumptions of Theorem~\ref{thm:minimax}, and the average payoffs along solutions are therefore at least non-negative. 
In fact, almost all (a set of full measure) initial conditions appear to converge to the pure equilibria at the boundary, with their time-average payoffs exceeding that of the Nash equilibrium and converging to the maximal welfare of 4, see Figure~\ref{FigPayoffs}.

\subsection{Asymmetric N-penny game}

Our second system is a system of $N$-player asymmetric mismatched pennies,
previously introduced in~\cite{kleinberg}.
There are three players, and  each  can choose between two  strategies: $\alpha_0$ and $\alpha_1$.
The payoffs for player $i$ with respect to player $i-1$ are given by the matrix
\begin{equation}
  A^{i-1, i}=
\begin{bmatrix}
0 & 1 \\
p & 0 
\end{bmatrix}.
\end{equation}
with $p>0$.

For odd $N$,
there is no Nash equilibrium in pure strategies.
In the replicator system, the pure strategy profiles are saddle-type stationary points of the ordinary differential equation,
linked by connecting orbits of mixed strategies.
The system has a unique
mixed Nash equilibrium defined by $x_i = \frac{1}{p+1}, \ i \in \{1,\dots, N\}$,
where each player obtains payoff of $\frac{p}{p+1}$.

\begin{figure}[t]
    \centering
        \includegraphics[width=42mm, trim= 40 0 40 0]{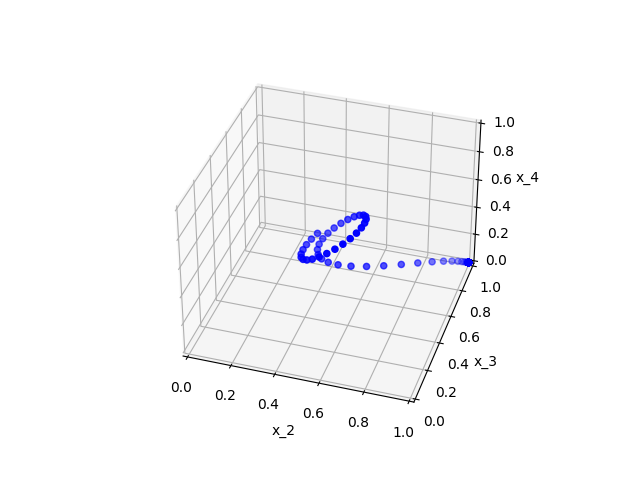}
  \includegraphics[width=42mm, trim= 40 0 40 0]{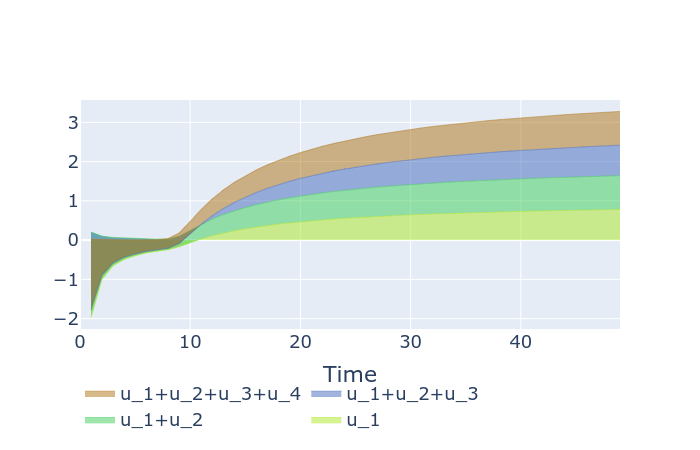}
    \includegraphics[width=42mm, trim= 40 0 40 0]{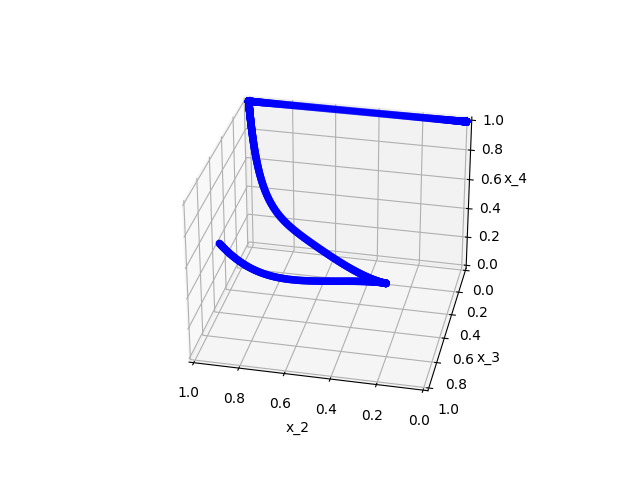} 
    \includegraphics[width=42mm, trim= 40 0 40 0]{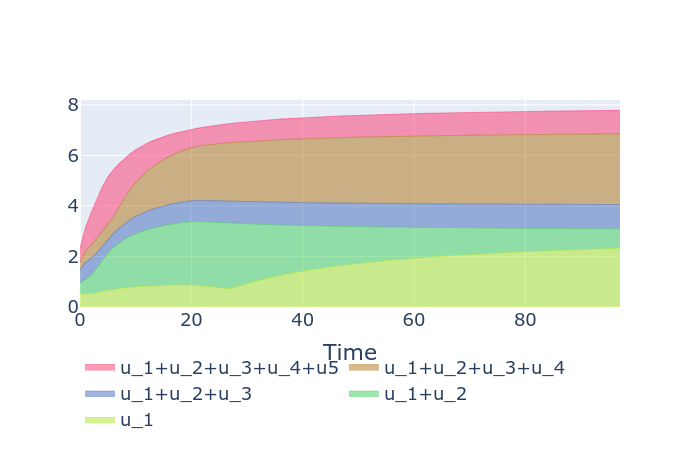}
    \caption{Time-average payoffs and social welfare of a sample learning trajectory in the matched-mismatched pennies game (top),
    and in the asymmetric 5-penny game with $p=3$
    (bottom, projection onto first three variables).}
    \label{FigPayoffs}
\end{figure}
The system was thoroughly analyzed in~\cite{kleinberg},
and the main result given
therein was that, for $N=3$ and $p>7$, all mixed strategies except for the diagonal 
converge to a sequence of orbits connecting boundary
stationary points. Moreover, the social welfare attained close to the boundary exceeds
the social welfare at the Nash equilibrium.
We extend these results.
From Theorem~\ref{thm:main} we deduce that, for all $N$ and for all $p \neq -1$,
the only limit sets in the interior are equilibria, periodic orbits, 
and cycles of connecting orbits to equilibria.
The payoff matrices satisfy the assumptions of Theorem~\ref{thm:minimax}, and, in particular, for all $p>0$, the mixed equilibrium yields the minimax payoff for each player, and time averages of payoffs along other orbits must exceed the minimax payoffs.
For almost all initial conditions,
the dynamics is attracted to the boundary cycle of average payoff $(p+1)\frac{N-1}{2}$ (see, e.g., Figure~\ref{FigPayoffs}), and indeed no chaotic emergent behavior  appears.

\subsection{A chaotic polymatrix replicator}

Our last system serves as a counterexample;
it shows that even in a binary three-player game, but without structured interactions (i.e., no cyclicity, all possible connections in the game graph), the learning trajectories of replicator dynamics can approach complex chaotic limit sets. 
The payoff matrices are given by
\begin{align}
\begin{aligned}
A^{1,1}&=
\begin{bmatrix}
\mu & 14 \\
0 & 0 
\end{bmatrix}, \quad
A^{2,1}=-A^{1,2}=
\begin{bmatrix}
-10 & 10 \\
0 & 0 
\end{bmatrix}, \\
A^{3,1}&=A^{3,2}=A^{3,3}=-A^{2,2}=
\begin{bmatrix}
-2 & 2 \\
0 & 0 
\end{bmatrix}, \\
A^{1,3}&=
\begin{bmatrix}
-25 & 29 \\
0 & 0 
\end{bmatrix},
\quad
A^{2,3}=
\begin{bmatrix}
0 & -11 \\
0 & 0 
\end{bmatrix}.
\end{aligned}
\end{align}
After some transformations (for details, see~\cite{Peixe}), we arrive at the following one-parameter system of differential equations:
\begin{align}
\begin{aligned}
\dot{x}_0 &= x_0 (1-x_0)(12-\mu +(\mu-14)x_0 -20x_1 -4x_2),\\
\dot{x}_1 &= x_1 (1-x_1)(-10+20x_0 +4x_1 -4x_2),\\
\dot{x}_2 &= x_2 (1-x_2)(27-54x_0+11x_1-4x_2),
\end{aligned}
\end{align}
where $x_i$ is the probability that player $i$ plays strategy $\alpha_0$, and $1-x_i$ is the probability that player $i$ plays $\alpha_1$.
This system was recently introduced by Peixe and Rodrigues~\cite{Peixe}, who formally showed by  combined theoretical and numerical approaches that the system contains a persistent strange (chaotic) attractor for a range of parameter values $\mu \in [1.4645,9.5055]$. We replicate their findings by integrating a sample trajectory and observing its approach to the chaotic attractor for $\mu=2.8$, see Figure~\ref{fig:chaos}. 
Due to lack of cyclicity, the game does not  guarantee the payoff structure given by Theorem~\ref{thm:minimax}.

\begin{figure}
\includegraphics[width=47mm, trim = 20 0 20 0]{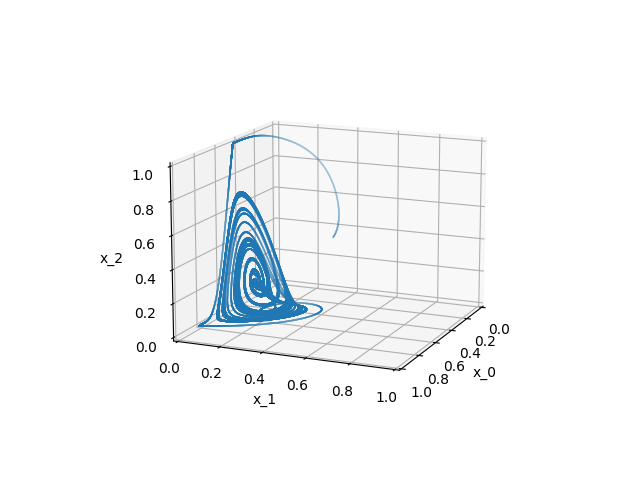}
    \includegraphics[width=35mm, trim =20 0 20 0]{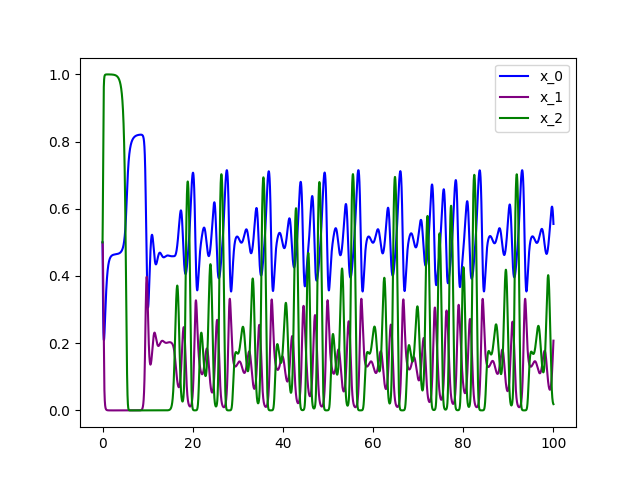}
    \caption{A learning trajectory approaching a chaotic attractor in the polymatrix replicator~\cite{Peixe} (left) and a plot of values of its coordinates (right). The game is characterized by unstructured interactions between payoffs and therefore breaks the assumptions of Theorems~\ref{thm:main} and~\ref{thm:main2}.}
    \label{fig:chaos}
\end{figure}
\balance
\section{Conclusions}

Numerous recent results regarding learning in games have established a clear separation between the idealized behavior of equilibration and the erratic, unpredictable, and typically chaotic behavior of learning dynamics even in simple games and domains.
At a first glance, this realization might seem to be a setback, but when viewed from the correct perspective it unveils  a new way of understanding learning dynamics, namely, by examining solution concepts from the topology of dynamical systems. 
Our results showcase the possibility of establishing links between the topological-combinatorial structure of multi-agent games (e.g., game graph, number of actions) to understand and constrain the topological complexity of game dynamics (Poincar\'e--Bendixson property) and finally link back to more traditional game theoretic analyses, such as calculating the efficiency of the system via social welfare. 
These connections showcase the promising advantages of this approach, 
which we hope will lead to more work along these lines in the future.

\nobalance{} 

\begin{acks}
AC was supported by the European Research Council (ERC) 
under the European Union's
Horizon 2020 research 
and innovation programme (grant agreement No.~758824 \textemdash INFLUENCE).
 \begin{figure}[h!]
   \centering
\includegraphics[width=0.4\columnwidth]{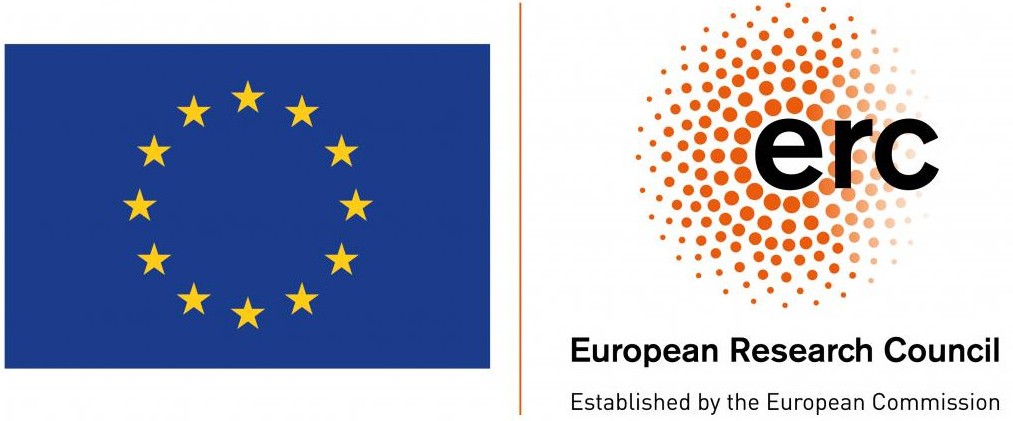}
\end{figure}

GP was supported in part by the National Research Foundation, Singapore under NRF 2018 Fellowship NRF-NRFF2018-07, AI Singapore Program (AISG Award No: AISG2-RP-2020-016), NRF2019-NRF-ANR095 ALIAS grant,  Technology and Research (A*STAR), AME Programmatic Fund (Grant No. A20H6b0151) from the Agency for Science, grant PIE-SGP-AI-2018-01 and Provost's Chair Professorship grant RGEPPV2101.

We would like to thank prof. Frans A. Oliehoek for his support, and helpful advice.
\end{acks}


\bibliographystyle{ACM-Reference-Format} 
\bibliography{bibliography}


\end{document}